\documentclass[envcountsect,envcountsame,11pt]{llncs}
\usepackage{amsmath,amssymb,bbm,graphicx}
\usepackage{enumerate}
\usepackage{epic,eepic,verbatim}
\usepackage{tikz}
\usepackage{rotating}
\def\proofbox{\hbox{\vrule height 1.4ex depth 0.2ex width 0.4em}}\def\qed{}

\tikzstyle{every picture} = [>=latex]
\usetikzlibrary{shapes}
\begin{document}

\title{Using Neighborhood Diversity to Solve Hard Problems}
\author{Robert~Ganian
\thanks{This research has been supported by the Czech research grant
202/11/0196.}}
\institute{
	Faculty of Informatics, Masaryk University\\
	Botanick\'a 68a, Brno, Czech Republic\\[.5ex]
\email{ganian@mail.muni.cz}}

\maketitle

\begin{abstract}
Parameterized algorithms are a very useful tool for dealing with NP-hard problems on graphs.
Yet, to properly utilize parameterized algorithms it is necessary to choose the right parameter based on the type of problem and properties of the target graph class.
Tree-width is an example of a very successful graph parameter, however it cannot be used on dense graph classes and there also exist problems which are hard even on graphs of bounded tree-width.
Such problems can be tackled by using vertex cover as a parameter, however this places severe restrictions on admissible graph classes.

Michael Lampis has recently introduced neighborhood diversity, a new graph parameter which generalizes vertex cover to dense graphs.
Among other results, he has shown that simple parameterized algorithms exist for a few problems on graphs of bounded neighborhood diversity.
Our article further studies this area and provides new algorithms parameterized by neighborhood diversity for the p-Vertex-Disjoint Paths, Graph Motif and Precoloring Extension problems -- the latter two being hard even on graphs of bounded tree-width.
\end{abstract}

\vspace{-0.8cm}

\section{Introduction}

\vspace{-0.2cm}
Parameterized algorithmics are a very successful approach to dealing with NP-hard problems on graphs.
The idea is that in real-life applications it is usually not necessary to solve problems on general graphs, but rather on graphs with some kind of structure present. 
It is then possible to use a structural parameter $k$ to describe this structure and use it to design parameterized algorithms which run in polynomial time as long as $k$ is bounded.

With respect to parameterized algorithms, we distinguish between those which run in time -- considerng a parameter $k$, an input of size $n$ and some computable function $f$ -- $O(f(k)\cdot n^{O(1)})$ and those with a runtime of $O(n^{f(k)})$.
Algorithms of the first type are called Fixed-Parameter Tractable (FPT in short) while those of the second type are called XP algorithms.
An FPT algorithm is typically much more practical than an XP one, and we will mostly be interested in obtaining FPT algorithms.
We refer to \cite{df99} for further information regarding parameterized complexity.

It is well known that FPT algorithms exist for a large number of NP-hard problems on graphs of bounded tree-width.
However, not all problems are efficiently solvable on graphs of bounded tree-width
and, in the context of parameterized algorithmics, vertex cover is often successfully used to solve these problems.
These include 
Precoloring Extension\cite{fgk10}, Equitable Coloring\cite{fgk10}, Equitable Connected Partition\cite{efg09} and Bandwidth \cite{flm08}, to name only a few.

Unfortunately, vertex cover only attains low values on severely restricted graph classes and so is only rarely applicable in practice.
It would be very useful to have a less restrictive parameter than vertex cover capable of solving such hard problems.
Neighborhood diversity, introduced by Michael Lampis in ESA 2010 \cite{lam10}, is a promising candidate for such a parameter.
Lampis' article mainly focuses on deciding $MS_1$ expressions on graphs of bounded vertex cover and neighborhood diversity, and has shown that both parameters may be used to decide $MS_1$ expressions in FPT time where the tower of exponents does not grow with quantifier alternations (unlike tree-width).
It also presented three simple algorithms for Hamiltonian Cycle, Graph Chromatic Number and Edge Dominating Set parameterized by neighborhood diversity (c.f. \cite[Theorem 6]{lam10}).

Our article further develops this direction and provides new algorithms parameterized by neighborhood diversity for the Graph Motif, p-Vertex-Disjoint Paths and Precoloring Extension problems.
We remark that while the problems solved in Lampis' original article may be solved in XP time on graphs of bounded rank-width (or clique-width), the problems considered here are much harder and cannot be efficiently solved on this graph class. In fact, two of the considered problems remain hard even on graphs of bounded tree-width.
Additionally, while FPT algorithms on vertex cover do exist for these problems, these cannot be straightforwardly transferred to neighborhood diversity due to the significant structural differences between these parameters.

\section{Vertex Cover and Neighborhood Diversity}
\begin{definition}
Given a graph $G=(V,E)$, a set $X\subseteq V$ is a vertex cover if every edge in $G$ is incident to at least one vertex in $X$.
\end{definition}

For readers not familiar with parameterized algorithm design, we briefly introduce the concept via some examples of algorithms for solving vertex cover parameterized by vertex cover.
Notice that designing an XP algorithm for this problem is trivial -- we may simply try selecting $k$ vertices from $V$ and see if they form a vertex cover. There are at most $|V|^k$ different possibilities for selecting the vertices, and checking whether they form a cover would take at most $|V|^2$ time. This gives a total runtime of $|V|^{k+2}$.

The problem can also be solved by a FPT algorithm.
We simply recursively repeat the following: first, delete all edges incident to a cover vertex, all cover vertices and all vertices of degree 0. We then choose any remaining vertex and branch out based on whether it is to be included in the cover or not.
If it is included in the cover, we have reduced the remaining number of cover vertices by 1, otherwise we know that all of its neighbors need to be cover vertices, reducing the remaining number of cover vertices by 1 or more.
Once $k$ cover vertices have been allocated in total, we simply check whether any edges remain uncovered.
And since each branching reduces the number of cover vertices by at least 1, we obtain a significantly improved total runtime of $2^k\cdot |V|^2$.

Finally, we remark that the $|V|^2$ time for checking whether $X$ is indeed a vertex cover can be easily improved to $k|V|$, and that the currently best known FPT algorithm for self-parameterized vertex cover in fact runs in time $O(1.2738^k+k|V|)$~\cite{cax10}.

We need one further notion before introducing neighborhood diversity. Note that for $v\in V$, $N(v)$ denotes the neighborhood of $v$.
\begin{definition}
We say that two vertices $v, v'$ of $G(V,E)$ have the same \emph{type} iff $N(v)\backslash \{v'\}=N(v')\backslash \{v\}$.
\end{definition}

The relation of having the same type is an equivalence. Furthermore, if a graph has vertex cover $k$, it cannot have more than $2^k+k$ equivalence classes of types (or type classes): if two non-cover vertices are adjacent to the same cover vertices, they belong to the same class (at most $2^k$ possible classes) and one class for each cover vertex ($k$ classes).

This type structure is very often utilized in parameters on vertex cover (see any of the parameterized algorithms referred to in the introduction).
The idea behind neighborhood diversity is based on this type structure.

\begin{definition}[\cite{lam10}]
A graph $G(V,E)$ has neighborhood diversity at most $w$, if there exists a
partition of $V$ into at most $w$ sets, such that all the vertices in each set have the same type.
\end{definition}

Notice that all vertices in a given type not only have the same neighborhood in $G$, but also form a clique or independent set in $G$. 
For algorithmic purposes, it is often useful to consider a \emph{type graph} $H$ of a graph $G$, where each vertex of $H$ is a type class in $G$ and two vertices in $H$ are adjacent iff there is a complete bipartite clique between these type classes in $G$.
It is not hard to see that the definitions force there to either be a complete bipartite clique or no edges between any two type classes.
The key property of graphs of bounded neighborhood diversity is that their type graphs have bounded size.


So, how does neighborhood diversity compare with other parameters?
As mentioned above, graphs of bounded vertex cover have bounded neighborhood diversity (although the bound may be single-exponential), but the opposite is not true, since large cliques have a neighborhood diversity of 1.
Neighborhood diversity is not comparable with tree-width -- large cliques and large trees being examples of graphs with one parameter bounded and the other unbounded.
Finally, graphs of bounded neighborhood diversity also have bounded clique-width and rank-width (see e.g. \cite{gh09} for a brief introduction to these measures).

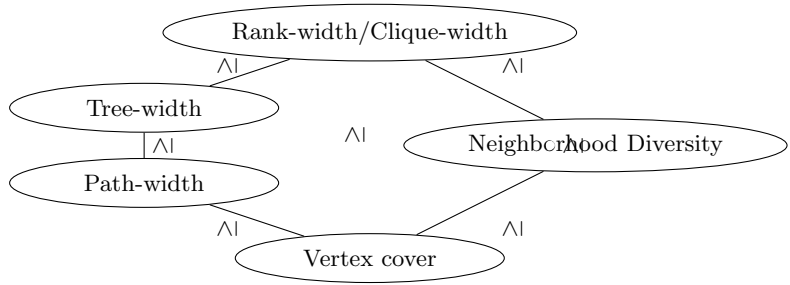
\begin{figure}[h]
\begin{center}
\begin{tikzpicture}[scale=1]
\tikzstyle{cloud} = [draw, ellipse, minimum height=2em, minimum width=11em]

\node (vc) at (0,-2)[cloud] {Vertex cover};
\node (kc) at (3,-0.5)[cloud] {Neighborhood Diversity};
\node (pw) at (-3,-1)[cloud] {Path-width};
\node (tw) at (-3,0)[cloud] {Tree-width};
\node (rw) at (0,1)[cloud] {Rank-width/Clique-width};

\draw (vc) -- (kc) -- (rw);
\draw (vc) -- (pw) -- (tw) -- (rw);
\draw (pw);

\tikzstyle{every node}=[]
\tikzstyle{empty}=[draw, shape=circle, minimum size=3pt,inner sep=0pt, fill=white, color=white]

\node (vk) at (2.2,-1.6)[empty, label=left:\begin{sideways}$\geq$\end{sideways}] {};
\node (lr) at (2.4,-0.5)[empty, label=right:\begin{sideways}$\geq$\end{sideways}] {};
\node (lr) at (2.2,0.56)[empty, label=left:\begin{sideways}$\geq$\end{sideways}] {};
\node (vp) at (-2.2,-1.6)[empty, label=right:\begin{sideways}$\geq$\end{sideways}] {};
\node (pt) at (-2.45,-0.5)[empty, label=left:\begin{sideways}$\geq$\end{sideways}] {};
\node (tr) at (-2.2,0.56)[empty, label=right:\begin{sideways}$\geq$\end{sideways}] {};
\node (pl) at (-0.5,-0.35)[empty, label=right:\begin{sideways}$\geq$\end{sideways}] {};

\end{tikzpicture}
\end{center}

\caption[center]{Relationships between selected graph parameters}
\end{figure}

The key result of Lampis' paper is an algorithm for $MS_1$ model checking on neighborhood diversity and vertex cover, where the height of the tower of exponents is fixed in the parameter (unlike $MS_1$ model checking on rank-width or $MS_2$ model checking on tree-width).
It is interesting that for $MS_1$ model checking, neighborhood diversity is actually exponentially faster than vertex cover, completely offsetting the worst case of neighborhood diversity being exponentially larger than vertex cover.

\begin{theorem}[{\cite[Theorem 2]{lam10}}]
There exists an algorithm which, given a graph $G$ with $l$ labels, neighborhood diversity
at most $w$ and an $MS_1$ formula $\phi$ with at most $q_S$ set variables and $q_V$ vertex variables, decides if $G \vDash \phi$ in time $2^{O(2^{q_S} (w+l)q^2_Sq_V log_{q_V}}) \cdot \phi$.
\end{theorem}

An important property of every graph parameter is its own computability, i.e. the time it takes to calculate it for a given input graph.
The greatest drawback of clique-width was the fact that it was actually NP-hard to compute it exactly, even when its value was bounded.
Rank-width, tree-width and vertex cover are all computable in FPT time when their respective values are bounded.
The nice structure of neighborhood diversity in fact allows us to go one step further and compute it in \emph{polynomial time} even if its value is large.

\begin{theorem}[{\cite[Theorem 5]{lam10}}]
There exists an algorithm which runs in polynomial time and given a graph $G(V,E)$ finds a
minimum partition of $V$ into neighborhood types.
\end{theorem}

\section{Algorithmic Results}
We will now present three new FPT algorithms parameterized by neighborhood diversity, illustrating various methods of exploiting the structure of such graphs to design efficient algorithms.
As mentioned earlier, it is generally not possible to translate algorithms from vertex cover to neighborhood diversity, since vertex-cover-parameterized algorithms strongly rely on exhaustively processing the cover vertices and then dealing with the remaining vertices, which are all independent.
On the other hand, here we have a relatively simple structure for all the vertices, but these are (often highly) connected to each other and there is no finite set of vertices to exhaustively process.

\subsection{Graph Motif}
\begin{definition}[Graph Motif]
\\
\noindent Input: A vertex-colored\footnote{This coloring need not be proper -- neighboring vertices may have the same color.} undirected graph $G$ and a multiset $M$ of colors.

\noindent Question: Does there exist a connected subgraph $B$ of $G$ such that the multiset of colors $col(B)$ occurring in $B$ is identical to $M$?
\end{definition}

\vspace{-0.2cm}

The Graph Motif problem has been introduced in \cite{lfs06} and arises naturally in bioinformatics, especially in
the context of metabolic network analysis.
Its complexity has been studied in \cite{ffvh11} and \cite{ab10}, the latter proving that the problem remains NP-hard even
on graphs of path-width 2.
Our first new result is an FPT algorithm for Graph Motif on graphs of bounded neighborhood diversity.

\begin{theorem}
The Graph Motif problem can be solved in time $O(2^{k}\!\cdot~\!\!\sqrt{|V|}|E|)$
on graphs of neighborhood diversity at most $k$.
\end{theorem}

\begin{proof}
The key observation lies in the fact that while there may be arbitrarily many possibilities for $B$ even on graphs of low neighborhood diversity, the highly connected structure enforced by our parameter allows us to work with a coarse representation of $B$ instead.

Let $B'$ denote the subgraph of $H$ (the type graph of $G$) induced by the types which contain at least a single vertex of $B$. 
Then the connectivity condition on $B$ immediately implies that $B'$ also needs to be connected.
We will proceed in two steps: First, we will consider all possibilities of $B'$, and for each $B'$ we will check whether there exists some admissible $B$.

For the first step, it suffices to realize that the number of vertices in $H$ is bounded by $k$ and so there are at most $2^k$ possible selections of $B'$.
For each considered $B'$, we move on to step two unless $B'$ is not connected (or if $B'$ contains only a single vertex of an independent set type if $|M|\geq 2$); this may be checked in advance for each $B'$.

We say $S$ is a skeleton of $B'$ if $S$ is an induced subgraph of $G$ such that
\begin{enumerate}
\item $col(S)\subseteq M$.
\item There exists an isomorphism $\phi$ between $S$ and $B'$.
\item $\forall v\in S: v\in \phi(v)$. (i.e. the isomorphism respects types)
\end{enumerate}
Before discussing how to deal with finding some or all possible skeletons, let us consider what happens once we have found a skeleton $S$.
Any vertex with a type in $B'$ is adjacent to some vertex in $S$, and in this way we may add any colors missing from $M$ to $S$ as long as they occur somewhere in $B'$.
Specifically, a $B$ inducing the selected $B'$ exists iff there exists at least one skeleton of $B'$ and $M$ is a subset of the multiset of colors occurring in $B'$.

Finally, what remains is to check whether there exists any skeleton of $B'$.
This final subproblem may be solved by simply finding a maximum matching between all colors in the multiset $M$
and all types in $B'$ -- with edges between colors and types containing that color.
Such a bipartite graph may be constructed in time $|E|$, has at most $|V|$ vertices and allows a maximum matching to be found in time $O(\sqrt{|V|}|E|)$. A skeleton exists iff the resulting maximum matching leaves no type unmatched.

So, the algorithm proceeds as follows.
First, it runs through all at most $2^k$ possible connected selections of $B'\subseteq H$.
Next, for each fixed $B'$ we check whether $M\subseteq col(B')$ -- if not, we move on to the next $B'$.
Otherwise, we check whether there exists a skeleton of $B'$ by the maximum matching subroutine.
If $M\subseteq col(B')$ and some $S$ is found, we can easily construct an admissible solution, and otherwise we know that no solution exists for the selected $B'$.
\qed
\end{proof}

\subsection{p-Vertex-Disjoint Paths}
There are two versions of the Vertex Disjoint Paths problem: either the number of paths is bounded by a constant, or it is part of the input and not bounded.
Both versions arise naturally in many fields.
The first version is much easier to solve and FPT algorithms are known even on general graphs \cite{rs95}.
The variant considered from now on is the second, more difficult one.

\begin{definition}[p-Vertex-Disjoint Paths]
\\
\noindent Input: An undirected graph $G=(V,E)$ and a set $P=\{s_i,t_i\}$, $s_i,t_i\in V$ of start and end vertices.

\noindent Question: Do there exist $|P|$ mutually vertex-disjoint paths $p_i$ between $s_i$ and $t_i$?
\end{definition}

Gurski and Wanke \cite{gw06} have shown that p-Vertex-Disjoint Paths is NP-hard on graphs of bounded rank-width and clique-width.
The problem is also known to be NP-hard on various other graph classes, such as planar graphs.
On the other hand, the problem admits an FPT algorithm for graphs of bounded tree-width.
We present an FPT algorithm for p-Vertex-Disjoint Paths on graphs of bounded neighborhood diversity, which allows the problem to be solved even on dense graphs of unbounded tree-width.
However, we will first need a few auxiliary results. 

\begin{definition}
We say that a p-Vertex-Disjoint Paths solution is \emph{simple} if each path only contains at most a single non-endpoint vertex of every type.
\end{definition}

\begin{lemma}
\label{lempath}
Any instance of p-Vertex-Disjoint Paths admits a solution iff it also admits a solution which is simple.
\end{lemma}
\begin{proof}
Assume that in our solution, there is a path of length at least four which contains several vertices of the same type.
Consider any type which occurs several times on the path, $x$ and $y$ being the first and last vertices of the same type occurring on the path and $z$ being the vertex after $y$ on the path.
Then the path may be rerouted from $x$ to $z$.
\qed
\end{proof}
\vspace{0.4cm}
We remark that the previous lemma is not true for edge-disjoint paths.
Next, we will make use of the deep result by Lenstra that integer linear programming is in fact FPT in the number of variables (\cite{len83}).
This provides a powerful tool for the design of FPT algorithms, which has unfortunately not had many interesting applications until recently (as noted in Niedermeier's monograph \cite{n06}).
The runtime has been later improved by Kannan \cite{kan87} and Frank and Tardos \cite{ft87}.
This approach was successfully used to design a number of FPT algorithms for graphs of bounded vertex cover \cite{flm08}.
The theorem we will use is:

\begin{theorem}[\cite{ft87,kan87,len83}]
\label{thmilp}
p-Variable Integer Linear Programming Feasiblity with an input of size $n$
can be solved in $O(q^{2.5q+o(q)}\cdot n)$ time and $n^{O(1)}$ space.
\end{theorem}

With these two results in hand, we are finally ready to solve the problem.

\begin{theorem}
The p-Vertex-Disjoint Paths problem may be solved in 
time $O(q^{2.5q+o(q)}\cdot n)$, where $q=k^2\cdot 2^k$,
on graphs of neighborhood diversity at most $k$.
\end{theorem}

\begin{proof}
Lemma \ref{lempath} tells us that we may in fact search for a solution where every path belongs to up to $k^2\cdot 2^k$ categories: $k^2$ possible start and end points, and less than $2^k$ possible routes (without the start and end points) in the type graph $H$.
For each route in $H$ we will discard those which are not valid paths starting and ending adjacent to the appropriate endpoints.
The problem is that there may still be many paths in each category, and we cannot exhaustively consider all possible assignments of paths into categories.

So, we construct an integer linear programming formulation capturing our problem.
Let $n_{a,b,i}$, $1\leq a,b\leq k$, $1\leq i\leq 2^k$ be integer variables describing the number of paths of each category.
We use constraints to ensure that:
\begin{enumerate}
\item The number of paths on the input which begin and end in a given type are equal to the sum of paths returned by the ILP program (e.g. $const_{a,b}=\sum_{1\leq i\leq 2^k} n_{a,b,i}$).
\item The number of paths crossing through every type in $H$ needs to be less or equal than the number of vertices of that type.
This is also straightforward, as the sum of certain predetermined variables needs to be less than a constant, with those paths ending as well as beginning in that type multiplied by two.
\end{enumerate}

If a solution exists, then the paths can be divided into categories as described above and the algorithm will return a satisfying evaluation.
On the other hand, if the algorithm returns a satisfying evaluation, we can easily transfer it into a solution by a greedy algorithm (all vertices in of any given type are mutually equivalent and so their choice does not matter as long as the correct route in $H$ is followed).
\qed
\end{proof}

\subsection{Precoloring Extension}
Precoloring Extension is a natural problem where we are given a partial proper coloring of a graph $G$ and the task is to extend it into a proper coloring of $G$ with $r$ colors.
The problem is W[1]-hard on graphs of bounded tree-width \cite{fetal11} and FPT when the vertex cover is bounded \cite{fgk10}.

As before, the algorithm of \cite{fgk10} cannot be directly applied to graphs of bounded neighborhood diversity.
Instead, we will once again construct an integer linear programming formulation.

\begin{theorem}
The Precoloring Extension problem can be solved in time
$O(q^{2.5q+o(q)}\cdot n)$, where $q=2^{2k}$,
on graphs of neighborhood diversity at most $k$.
\end{theorem}

\begin{proof}
Recall that every type class either form an independent set or clique.
It is a simple observation that if any vertex in an independent set type is precolored, we may always use that color for all the other vertices of the same type.
On the other hand, if no vertex in an independent set type is precolored, any color which is used for one vertex of that type can be used for all the others.
So it suffices to only consider inputs where independent set types are fully precolored or only contain a single vertex.

We will use similar categories as those in the previous algorithm, however this time these will have to be finer.
A color category will be the set of colors which appear precolored in the same types, and these are fixed by the input.
For each color category, we will then distinguish color subcategories:
Two colors are in the same subcategory if they are in the same category and appear (colored or precolored) in the same types.
Obviously, if a proper extension of a precoloring exists, then all colors may be divided into subcategories accordingly.
The number of color subcategories is at most $2^{2k}$, and since the coloring needs to be proper we only consider admissible subcategories (omit those containing any adjacent types).

So, we use variables to denote the number of colors in each category subcategory (for instance $n_{a,b}$ for category $a$ and subcategory $b$).
The key observation used to build our constraints is that colors in the same subcategories not only appear in the same types, they also color exactly one vertex in each such type.
The constraints are:
\begin{enumerate}
\item The number of colors in each category (follows from input) equals the sum of all subcategories in that category.
\item The number of vertices of each type equals the sum of all subcategories which color a vertex of that type. (ensuring that every vertex has been colored)
\end{enumerate}

As in the previous algorithm, if a solution exists, then the colors can be divided into subcategories as described above and the algorithm will return a satisfying evaluation.
On the other hand, if the algorithm returns a satisfying evaluation, we can immediately use it to construct a proper coloring.
\qed
\end{proof}

\section{Concluding Remarks}
We have presented three novel parameterized algorithms on neighborhood diversity, and shown that it indeed is a parameter capable of solving various difficult problems - even problems which are hard on graphs of bounded tree-width.
The applied methods may also be helpful for designing other such algorithms on neighborhood diversity.
However, there is still a lot of work to be done.
Some problems, such as the graph layout problems studied in \cite{flm08}, require a different approach if an FPT algorithm is to be found.

Additionally, while neighborhood diversity can attain low values even on graphs of high vertex cover (even on those of high tree-width in fact), it can also be up to single-exponentially larger than vertex cover.
This unfortunately means that even if a problem can be solved by an FPT algorithm on neighborhood diversity, it may sometimes be more efficient to in fact use vertex cover and not neighborhood diversity.
This is unfortunate, and it remains an open question whether another generalization of vertex cover could be found which could not be significantly larger than vertex cover.
\bibliographystyle{abbrv}

\bibliography{gtbibsofsem12}

\begin{thebibliography}{10}

\bibitem{ab10}
A.~M. Ambalath, R.~Balasundaram, C.~R. H, V.~Koppula, N.~Misra, G.~Philip, and
  M.~S. Ramanujan.
\newblock On the kernelization complexity of colorful motifs.
\newblock In {\em IPEC 2010}, number 6478 in LNCS, pages 14--25, 2010.

\bibitem{cax10}
J.~Chen, I.~A. Kanj, and G.~Xia.
\newblock Improved upper bounds for vertex cover.
\newblock {\em Theor. Comput. Sci.}, 411:3736--3756, 2010.

\bibitem{df99}
R.~G. Downey and M.~R. Fellows.
\newblock {\em Parameterized complexity}.
\newblock Monographs in Computer Science. Springer, 1999.

\bibitem{efg09}
R.~Enciso, M.~R. Fellows, J.~Guo, I.~Kanj, F.~Rosamond, and O.~Such\'{y}.
\newblock What makes equitable connected partition easy.
\newblock In J.~Chen and F.~V. Fomin, editors, {\em Parameterized and Exact
  Computation}, pages 122--133. Springer-Verlag, Berlin, Heidelberg, 2009.

\bibitem{ffvh11}
M.~R. Fellows, G.~Fertin, D.~Hermelin, and S.~Vialette.
\newblock Upper and lower bounds for finding connected motifs in vertex-colored
  graphs.
\newblock {\em J. Comput. Syst. Sci.}, 77:799--811, July 2011.

\bibitem{fetal11}
M.~R. Fellows, F.~V. Fomin, D.~Lokshtanov, F.~Rosamond, S.~Saurabh, S.~Szeider,
  and C.~Thomassen.
\newblock On the complexity of some colorful problems parameterized by
  treewidth.
\newblock {\em Inf. Comput.}, 209:143--153, February 2011.

\bibitem{flm08}
M.~R. Fellows, D.~Lokshtanov, N.~Misra, F.~A. Rosamond, and S.~Saurabh.
\newblock Graph layout problems parameterized by vertex cover.
\newblock In {\em Proceedings of the 19th International Symposium on Algorithms
  and Computation}, pages 294--305, Berlin, Heidelberg, 2008. Springer-Verlag.

\bibitem{fgk10}
J.~Fiala, P.~A. Golovach, and J.~Kratochv\'{\i}l.
\newblock Parameterized complexity of coloring problems: Treewidth versus
  vertex cover.
\newblock {\em Theoretical Computer Science}, In Press, 2010.

\bibitem{ft87}
A.~Frank and {\'E}.~Tardos.
\newblock An application of simultaneous diophantine approximation in
  combinatorial optimization.
\newblock {\em Combinatorica}, 7(1):49--65, 1987.

\bibitem{gh09}
R.~Ganian and P.~Hlin\v{e}n\'y.
\newblock On parse trees and {M}yhill--{N}erode--type tools for handling graphs
  of bounded rank-width.
\newblock {\em Discrete Appl. Math.}, 2009.
\newblock To appear.

\bibitem{gw06}
F.~Gurski and E.~Wanke.
\newblock Vertex disjoint paths on clique-width bounded graphs.
\newblock {\em Theor. Comput. Sci.}, 359(1-3):188--199, 2006.

\bibitem{kan87}
R.~Kannan.
\newblock Minkowski's convex body theorem and integer programming.
\newblock {\em Math. Oper. Res.}, 12:415--440, 1987.

\bibitem{lfs06}
V.~Lacroix, C.~G. Fernandes, and M.-F.~F. Sagot.
\newblock {Motif search in graphs: application to metabolic networks.}
\newblock {\em IEEE/ACM transactions on computational biology and
  bioinformatics / IEEE, ACM}, 3(4):360--368, October 2006.

\bibitem{lam10}
M.~Lampis.
\newblock Algorithmic meta-theorems for restrictions of treewidth.
\newblock In {\em ESA (1)}, pages 549--560, 2010.

\bibitem{len83}
H.~Lenstra.
\newblock Integer programming with a fixed number of variables.
\newblock {\em Math. Oper. Res.}, 8:538--548, 1983.

\bibitem{n06}
R.~Niedermeier.
\newblock {\em Invitation to Fixed-Parameter Algorithms}.
\newblock xford Lecture Series in Mathematicsand Its Applications. Oxford
  University Press, 2006.

\bibitem{rs95}
N.~Robertson and P.~D. Seymour.
\newblock Graph minors .xiii. the disjoint paths problem.
\newblock {\em J. Comb. Theory, Ser. B}, 63(1):65--110, 1995.

\end{thebibliography}

\end{document}